\documentclass[11pt]{article}%
\usepackage{amsfonts}
\usepackage{amsmath}
\usepackage{amssymb}
\usepackage{graphicx}%
\newtheorem{theorem}{Theorem}
\newtheorem{acknowledgement}[theorem]{Acknowledgement}

\newtheorem{definition}[theorem]{Definition}

\newtheorem{lemma}[theorem]{Lemma}
\newtheorem{notation}[theorem]{Notation}

\newtheorem{proposition}[theorem]{Proposition}

\newenvironment{proof}[1][Proof]{\noindent\textbf{#1.} }{\ \rule{0.5em}{0.5em}}
\begin{document}

\title{A Deformation Quantization Theory for Non-Commutative Quantum Mechanics }
\author{Nuno Costa Dias \ \ \ Maurice de Gosson
\and Franz Luef \ \ \ Jo\~{a}o Nuno Prata}
\maketitle

\begin{abstract}
We show that the deformation quantization of non-commutative quantum mechanics
previously considered by Dias and Prata can be expressed as a Weyl calculus on
a double phase space. We study the properties of the star-product thus
defined, and prove a spectral theorem for the star-genvalue equation using an
extension of the methods recently initiated by de Gosson and Luef.

\end{abstract}

\textbf{Mathematics Subject Classification\ (2000)}: 47G30, 81S10


\section{Introduction}

The generalization of quantum mechanics obtained by considering canonical
extensions of the Heisenberg algebra is usually referred to as non-commutative
quantum mechanics (NCQM), a theory that displays an additional non-commutative
structure in the configurational and momentum sectors. One of the main
incentives for studying NCQM comes from the quest for a theory of quantum
gravity. It is widely expected that such a theory will determine a
modification of the structure of space-time of some non-commutative nature
\cite{Delduc,Douglas,Seiberg,Szabo}. Hence, deviations from the predictions of
standard quantum mechanics, and particularly those arising from considering
its non-commutative extensions, could be regarded as a sign of the underlying
theory of quantum gravity. In Dias and Prata \cite{dipra1,badipr} two of us
have discussed various aspects of NCQM related to Flato--Sternheimer
deformation quantization \cite{BFFLS1,BFFLS2}. In this paper we propose an
operator theoretical approach, based on previous work de Gosson and Luef
\cite{GOLU1} (the remaining two of us). In that article it was shown that the
Moyal--Groenewold product $a\star b$ of two functions on $\mathbb{R}^{2n}$ can
be interpreted in terms of a Weyl calculus on $\mathbb{R}^{2n}$. In fact,
\begin{equation}
a\star b=\widetilde{A}b \label{1}%
\end{equation}
where $\widetilde{A}$ is the phase space operator with Weyl symbol
$\widetilde{a}$ defined on $\mathbb{R}^{2n}\oplus\mathbb{R}^{2n}$ by
\begin{equation}
\widetilde{a}(z,\zeta)=a(z-\tfrac{1}{2}J\zeta) \label{3}%
\end{equation}
($J$ is the standard symplectic matrix).

In this paper we show that this redefinition of the starproduct can be
modified so that it leads to a natural notion of deformation quantization for
the NCQM associated with an antisymmetric matrix of the type
\begin{equation}
\Omega=%
\begin{pmatrix}
\hbar^{-1}\Theta & I\\
-I & \hbar^{-1}N
\end{pmatrix}
\label{omega}%
\end{equation}
where $\Theta$ and $N$ measure the non-commutativity in the position and
momentum variables, respectively. We define a new starproduct $\star_{\Omega}$
by replacing formula (\ref{1}) by
\begin{equation}
a\star_{\Omega}b=\widetilde{A}_{\Omega}b \label{astarom}%
\end{equation}
where $\widetilde{A}_{\Omega}$ is the operator with Weyl symbol
\begin{equation}
\widetilde{a}_{\Omega}(z,\zeta)=a(z-\tfrac{1}{2}\Omega\zeta). \label{atild1}%
\end{equation}
Of course (\ref{atild1}) reduces to (\ref{1}) when $\Theta=N=0$.

In this article we are going to rigorously justify the definition above and
study the properties of this new starproduct $\star_{\Omega}$. The difficulty
associated with the fact that the symplectic form associated with $\Omega$
depends on $\hbar$ will be resolved (we will show that $a\star_{\Omega}b$ is
well-defined as a starproduct thanks to supplementary conditions on $\Theta$
and $N$ which are physically meaningful). In fact \ $\star_{\Omega}$ coincides
with the starproduct defined (in terms of the generalized Weyl-Wigner map
\cite{Dias1}) in Eqn. (21) of \cite{dipra1}, and where it was shown that it is
related to the standard starproduct. (In the same paper it was concluded in
Eqn. (53) that the generalized starproduct between two polynomials can be
represented as a kind of \textquotedblleft Bopp shift\textquotedblright\ which
also turns out to be identical with $\star_{\Omega}$).

\begin{notation}
The generic point of phase space $\mathbb{R}^{2n}$ is denoted $z=(x,p)$. We
denote by $\operatorname*{Sp}(2n,\mathbb{R})$ the standard symplectic group,
defined as the group of linear automorphisms of $\mathbb{R}^{2n}$ equipped
with the symplectic form $\sigma(z,z^{\prime})=Jz\cdot z^{\prime}$, $J=%
\begin{pmatrix}
0 & I\\
-I & 0
\end{pmatrix}
$. We use the standard notation $\mathcal{S}(\mathbb{R}^{m})$ and
$\mathcal{S}^{\prime}(\mathbb{R}^{m})$ for the Schwartz space of test
functions on $\mathbb{R}^{m}$ and its dual.
\end{notation}

\section{Description Of the Problem}

Let us begin by explaining what we mean by non-commutativity in the present
context. The study of non-commutative field theories and their connections
with quantum gravity (see \cite{Binz,Delduc,Douglas,Seiberg,Szabo} and the
references therein) leads to the consideration of commutation relations of the
type
\begin{equation}
\lbrack\widetilde{z}_{\alpha},\widetilde{z}_{\beta}]=i\hbar\omega_{\alpha
\beta}\text{ ,\ }1\leq\alpha,\beta\leq2n \label{CCR2}%
\end{equation}
where $\Omega=(\omega_{\alpha\beta})_{1\leq\alpha,\beta\leq2n}$ is the
$2n\times2n$ antisymmetric matrix defined by (\ref{omega}) where
$\Theta=(\theta_{\alpha\beta})_{1\leq\alpha,\beta\leq n}$ and $N=(\eta
_{\alpha\beta})_{1\leq\alpha,\beta\leq n}$ are antisymmetric matrices
measuring the non-commutativity in the position and momentum variables. We
have set here $\widetilde{z}_{\alpha}=\widetilde{x}_{\alpha}$ if $1\leq
\alpha\leq n$ and $\widetilde{z}_{\alpha}=\widetilde{p}_{\alpha-n}$ if
$n+1\leq\alpha\leq2n$, where
\begin{align}
\widetilde{x}_{\alpha}  &  =x_{\alpha}+\tfrac{1}{2}i\sum\nolimits_{\beta
}\theta_{\alpha\beta}\partial_{x_{\beta}}+\tfrac{1}{2}i\hbar\partial
_{p_{\alpha}}\label{ccr31}\\
\widetilde{p}_{\alpha}  &  =p_{\alpha}-\tfrac{1}{2}i\hbar\partial_{x_{\alpha}%
}+\tfrac{1}{2}i\sum\nolimits_{\beta}\eta_{\alpha\beta}\partial_{p_{\beta}}.
\label{ccr32}%
\end{align}
It turns out that, as proved in \cite{dipra1}, $\Omega$ is invertible if
\begin{equation}
\theta_{\alpha\beta}\eta_{\gamma\delta}<\hbar^{2}\text{ \ for \ }1\leq
\alpha<\beta\leq n\text{ and }1\leq\gamma<\delta\leq n.\text{ } \label{det}%
\end{equation}
We will assume from now on that these conditions are satisfied; that this
requirement is physically meaningful is well-known (it is fulfilled for
instance in the case of the non-commutative quantum well; see for instance
\cite{bracz,cahakola}). Since we will be concerned with a deformation
quantization with parameter $\hbar\rightarrow0$ we will furthermore assume
that $\Theta$ and $N$ depend smoothly on $\hbar$ in such a way that
\begin{equation}
\Theta(\hbar)=o(\hbar^{2})\text{ \ and }N(\hbar)=o(\hbar^{2}) \label{oh}%
\end{equation}
(recall that $f(\hbar)=o(\hbar^{m})$ means that $\lim_{\hbar\rightarrow
0}(f(\hbar)/\hbar^{m})=0$). We thus have
\[
\lim_{\hbar\rightarrow0}\Omega=J=%
\begin{pmatrix}
0 & I\\
-I & 0
\end{pmatrix}
\]
(the standard symplectic matrix). It turns out that the conditions (\ref{oh})
are compatible with numerical results in \cite{bracz,cahakola} where it is
shown that the estimates $\theta\leq4\times10^{-40}$ $m^{2}$ and $\eta
\leq1.76\times10^{-61}$ $kg^{2}m^{2}s^{-2}$ hold. Moreover, the analysis of
non-commutative quantum mechanics in the context of dissipative open systems,
reveals that a transition $\theta\rightarrow0$ occurs prior to $\hbar
\rightarrow0$ \cite{Dias2}.

These facts, and the theory developed in \cite{GOLU1}, suggests that we
represent $\widetilde{z}=(\widetilde{z}_{1},...,\widetilde{z}_{2n})$ by the
vector operator%
\begin{equation}
\widetilde{z}=z+\tfrac{1}{2}i\hbar\Omega\partial_{z} \label{cc4}%
\end{equation}
which acts on functions defined on the phase space $\mathbb{R}^{2n}$. Notice
that the conditions (\ref{oh}) show that in the limit $\hbar\rightarrow0$ we
have the asymptotic formulae
\begin{equation}
\widetilde{x}_{\alpha}=x_{\alpha}+\tfrac{1}{2}i\hbar\partial_{p_{\alpha}%
}+o(\hbar^{2})\text{ \ , \ }\widetilde{p}_{\alpha}=p_{\alpha}-\tfrac{1}%
{2}i\hbar\partial_{x_{\alpha}}+o(\hbar^{2}). \label{asy1}%
\end{equation}
The \textquotedblleft quantization rules\textquotedblright\ (\ref{cc4}) lead
us to the consideration of pseudo-differential operators formally defined by
(\ref{atild1}).

The underlying symplectic structure we are going to use is defined as follows.
We will denote by $s$ a linear automorphism of $\mathbb{R}^{2n}$ such that
$\sigma=s^{\ast}\omega$; equivalently $sJs^{T}=\Omega$. Thus $s$ is a
symplectomorphism $s:(\mathbb{R}^{2n},\sigma)\longrightarrow(\mathbb{R}%
^{2n},\omega)$. Note that the mapping $s$ is sometimes called the
\textquotedblleft Seiberg--Witten map\textquotedblright\ in the physical
literature; its existence is of course mathematically a triviality (because it
is just a linear version of Darboux's theorem, see \cite{Birk}, \S 1.1.2).
Writing $s$ in block-matrix form $%
\begin{pmatrix}
A & B\\
C & D
\end{pmatrix}
$ the condition $sJs^{T}=\Omega$ is equivalent to
\[
AB^{T}-BA^{T}=\hbar^{-1}\Theta\text{ \ , \ }CD^{T}-DC^{T}=\hbar^{-1}N\text{
\ , \ }AD^{T}-BC^{T}=I.
\]
Of course, the automorphism $s$ is not uniquely defined: if $s^{\ast}%
\omega=s^{\prime\ast}\omega$ then $s^{-1}s^{\prime}\in\operatorname*{Sp}%
(2n,\mathbb{R})$. Also note that in the limit $\hbar\rightarrow0$ the matrices
$\hbar^{-1}\Theta$ and $\hbar^{-1}N$ vanish and $s$ becomes, as expected,
symplectic in the usual sense, that is $s\in\operatorname*{Sp}(2n,\mathbb{R})$.

\section{Definition of the starproduct $\star_{\Omega}$}

Let $\omega$ be the symplectic form on $\mathbb{R}^{2n}$ defined by
$\omega(z,z^{\prime})=z\cdot\Omega^{-1}z^{\prime}$; it coincides with the
standard symplectic form $\sigma$ when $\Omega=J$.

We will need the two following unitary transformations:

\begin{itemize}
\item The $\Omega$-symplectic transform $F_{\Omega}$ defined, for
$a\in\mathcal{S}(\mathbb{R}^{2n})$, by
\begin{equation}
F_{\Omega}a(z)=\left(  \tfrac{1}{2\pi\hbar}\right)  ^{n}|\det\Omega
|^{-1/2}\int_{\mathbb{R}^{2n}}e^{-\frac{i}{\hbar}\omega(z,z^{\prime}%
)}a(z^{\prime})dz^{\prime}; \label{aoum}%
\end{equation}
it extends into an involutive automorphism of $\mathcal{S}^{\prime}%
(\mathbb{R}^{2n})$ (also denoted by $F_{\Omega}$) and whose restriction to
$L^{2}(\mathbb{R}^{2n})$ is unitary;

\item The unitary operator $\widetilde{T}_{\Omega}(z_{0})$ defined, for
$\Psi\in\mathcal{S}^{\prime}(\mathbb{R}^{2n})$ by the formula
\begin{equation}
\widetilde{T}_{\Omega}(z_{0})\Psi(z)=e^{-\frac{i}{\hbar}\omega(z,z_{0})}%
\Psi(z-\tfrac{1}{2}z_{0}). \label{hwbis}%
\end{equation}
Notice that when $\Omega=J$ we have $\widetilde{T}_{\Omega}(z_{0}%
)=\widetilde{T}(z_{0})$ where $\widetilde{T}(z_{0})$ is defined by formula (8)
in \cite{GOLU1}.
\end{itemize}

Let us express the operator $\widetilde{A}_{\Omega}=a(z+\tfrac{1}{2}%
i\hbar\Omega\partial_{z})$ in terms of $F_{\Omega}a$ and $\widetilde
{T}_{\Omega}(z_{0})$.

\begin{proposition}
Let $\widetilde{A}_{\Omega}$ be the operator on $\mathbb{R}^{2n}$ with Weyl
symbol
\begin{equation}
\widetilde{a}_{\Omega}(z,\zeta)=a(z-\tfrac{1}{2}\Omega\zeta). \label{atilde}%
\end{equation}
We have%
\begin{equation}
\widetilde{A}_{\Omega}=\left(  \tfrac{1}{2\pi\hbar}\right)  ^{n}|\det
\Omega|^{-1/2}\int_{\mathbb{R}^{2n}}F_{\Omega}a(z)\widetilde{T}_{\Omega}(z)dz.
\label{aomeg}%
\end{equation}

\end{proposition}

\begin{proof}
Let us denote by $\widetilde{B}$ the right-hand side of (\ref{aomeg}). We
have, setting $u=z-\frac{1}{2}z_{0}$,%
\begin{align*}
\widetilde{B}\Psi(z)  &  =\left(  \tfrac{1}{2\pi\hbar}\right)  ^{n}|\det
\Omega|^{-1/2}\int_{\mathbb{R}^{2n}}F_{\Omega}a(z_{0})e^{-\frac{i}{\hbar
}\omega(z,z_{0})}\Psi(z-\tfrac{1}{2}z_{0})dz_{0}\\
&  =\left(  \tfrac{2}{\pi\hbar}\right)  ^{n}|\det\Omega|^{-1/2}\int
_{\mathbb{R}^{2n}}F_{\Omega}a[2(z-u)]e^{\frac{2i}{\hbar}\omega(z,u)}\Psi(u)du
\end{align*}
hence the kernel of $\widetilde{B}$ is given by%
\[
K(z,u)=\left(  \tfrac{2}{\pi\hbar}\right)  ^{n}|\det\Omega|^{-1/2}F_{\Omega
}a[2(z-u)]e^{\frac{2i}{\hbar}\omega(z,u)}.
\]
It follows that the Weyl symbol $\widetilde{b}$ of $\widetilde{B}$ is given by%
\begin{align*}
\widetilde{b}(z,\zeta)  &  =\int_{\mathbb{R}^{2n}}e^{-\frac{i}{\hbar}%
\zeta\cdot\zeta^{\prime}}K(z+\tfrac{1}{2}\zeta^{\prime},z-\tfrac{1}{2}%
\zeta^{\prime})d\zeta^{\prime}\\
&  =\left(  \tfrac{2}{\pi\hbar}\right)  ^{n}|\det\Omega|^{-1/2}\int
_{\mathbb{R}^{2n}}e^{-\frac{i}{\hbar}\zeta\cdot\zeta^{\prime}}F_{\Omega
}a(2\zeta^{\prime})e^{-\frac{2i}{\hbar}\omega(z,\zeta^{\prime})}d\zeta
^{\prime}%
\end{align*}
that is, using the obvious relation
\[
\zeta\cdot\zeta^{\prime}+2\omega(z,\zeta^{\prime})=\omega(2z-\Omega\zeta
,\zeta^{\prime})
\]
together with the change of variables $z^{\prime}=2\zeta^{\prime}$,
\begin{align*}
\widetilde{b}(z,\zeta)  &  =\left(  \tfrac{2}{\pi\hbar}\right)  ^{n}%
|\det\Omega|^{-1/2}\int_{\mathbb{R}^{2n}}e^{-\frac{i}{\hbar}\omega
(2z-\Omega\zeta,\zeta^{\prime})}F_{\Omega}a(2\zeta^{\prime})d\zeta^{\prime}\\
&  =\left(  \tfrac{1}{2\pi\hbar}\right)  ^{n}|\det\Omega|^{-1/2}%
\int_{\mathbb{R}^{2n}}e^{-\frac{i}{\hbar}\omega(z-\frac{1}{2}\Omega
\zeta,z^{\prime})}F_{\Omega}a(z^{\prime})dz^{\prime}%
\end{align*}
that is, using the fact that $F_{\Omega}F_{\Omega}$ is the identity,%
\[
\widetilde{b}(z,\zeta)=a(z-\tfrac{1}{2}\Omega\zeta)=\widetilde{a}_{\Omega
}(z,\zeta)
\]
which concludes the proof.
\end{proof}

The result above motivates the following definition:

\begin{definition}
Let $a\in\mathcal{S}^{\prime}(\mathbb{R}^{2n})$ and $b\in\mathcal{S}%
(\mathbb{R}^{2n})$ . The $\Omega$-starproduct of $a$ and $b$ is the element of
$\mathcal{S}^{\prime}(\mathbb{R}^{2n})$ defined by%
\begin{equation}
a\star_{\Omega}b=\widetilde{A}_{\Omega}b. \label{defstar}%
\end{equation}

\end{definition}

Note that it is not yet clear from the definition above that $\star_{\Omega}$
is a bona fide starproduct. For instance, while it is obvious that
$1\star_{\Omega}b=b$ (because the operator $\widetilde{A}_{\Omega}$ with
symbol $a=1$ is the identity), the formula $b\star_{\Omega}1=b$ is certainly
not, and it is even less clear that $\star_{\Omega}$ is associative!

\section{A New Star-Product Is Born...}

It turns out that we can reduce the study of the newly defined starproduct to
that of the usual Groenewold--Moyal product $\star$. For this we will need
Lemma \ref{lemma3} below.

\begin{lemma}
\label{lemma3}Let $s$ be a linear automorphism such that $\sigma=s^{\ast
}\omega$ and define a automorphism $M_{s}:\mathcal{S}^{\prime}(\mathbb{R}%
^{2n})\longrightarrow\mathcal{S}^{\prime}(\mathbb{R}^{2n})$ by%
\begin{equation}
M_{s}\Psi(z)=\sqrt{|\det s|}\Psi(sz). \label{ms}%
\end{equation}
(hence $M_{s}$ is unitary on $L^{2}(\mathbb{R}^{2n})$). We have
\begin{equation}
M_{s}\widetilde{A}_{\Omega}=\widetilde{A^{\prime}}M_{s} \label{intera}%
\end{equation}
\ where $\widetilde{A^{\prime}}=\widetilde{A^{\prime}}_{J}$ corresponds to the
operator $\widehat{A^{\prime}}$ acting on $L^{2}(\mathbb{R}^{n})$ with Weyl
symbol $a^{\prime}(z)=a(sz)$, and hence
\begin{equation}
\text{\ }M_{s}(a\star_{\Omega}b)=\sqrt{|\det s|}(a^{\prime}\star b^{\prime})
\label{inter1}%
\end{equation}
where $b^{\prime}(z)=b(sz)$.
\end{lemma}

\begin{proof}
Formula (\ref{inter1}) immediately follows from formula (\ref{intera}). To
prove formula (\ref{intera}) one first checks the identities
\[
M_{s}\widetilde{T}_{\Omega}(z_{0})=\widetilde{T}(s^{-1}z_{0})M_{s}\text{
\ ,\ \ }M_{s}F_{\Omega}=F_{J}M_{s}%
\]
(the verification of which is purely computational and therefore left to the
reader); using these identities we have%
\begin{align*}
M_{s}\widetilde{A}_{\Omega}  &  =\left(  \tfrac{1}{2\pi\hbar}\right)
^{n}|\det\Omega|^{-1/2}\int_{\mathbb{R}^{2n}}F_{\Omega}a(z_{0})\widetilde
{T}(s^{-1}z_{0})M_{s}dz_{0}\\
&  =\left(  \tfrac{1}{2\pi\hbar}\right)  ^{n}|\det s||\det\Omega|^{-1/2}%
\int_{\mathbb{R}^{2n}}F_{\Omega}a(sz)\widetilde{T}(z)M_{s}dz\\
&  =\left(  \tfrac{1}{2\pi\hbar}\right)  ^{n}|\det s|^{1/2}|\det\Omega
|^{-1/2}\int_{\mathbb{R}^{2n}}M_{s}F_{\Omega}a(z)\widetilde{T}(z)M_{s}dz\\
&  =\left(  \tfrac{1}{2\pi\hbar}\right)  ^{n}|\det s|^{1/2}|\det\Omega
|^{-1/2}\int_{\mathbb{R}^{2n}}F_{J}(M_{s}a)(z)\widetilde{T}(z)M_{s}dz\\
&  =\widetilde{A^{\prime}}M_{s}.
\end{align*}

\end{proof}

The double equality
\begin{equation}
1\star_{\Omega}a=a\star_{\Omega}1=a \label{unital}%
\end{equation}
now immediately follows from formula (\ref{inter1}): we have
\[
M_{s}(1\star_{\Omega}a)=\sqrt{|\det s|}(1\star a^{\prime})=1\star M_{s}%
a=M_{s}a
\]
hence we recover the equality $1\star_{\Omega}a=a$; similarly
\[
M_{s}(a\star_{\Omega}1)=\sqrt{|\det s|}(a^{\prime}\star1)=M_{s}a\star1=M_{s}a
\]
hence $a\star_{\Omega}1=a$.

Let us now prove the associativity of the $\Omega$-starproduct:

\begin{proposition}
Assume that the starproducts $a\star_{\Omega}(b\star_{\Omega}c)$ and
$(a\star_{\Omega}b)\star_{\Omega}c$ are defined. We then have
\begin{equation}
a\star_{\Omega}(b\star_{\Omega}c)=(a\star_{\Omega}b)\star_{\Omega}c.
\label{assoc}%
\end{equation}

\end{proposition}

\begin{proof}
It is of course sufficient to show that%
\begin{equation}
M_{s}\left[  a\star_{\Omega}(b\star_{\Omega}c)\right]  =M_{s}\left[
(a\star_{\Omega}b)\star_{\Omega}c)\right]  . \label{massoci}%
\end{equation}
We have, by repeated use of (\ref{inter1}) together with the definition of
$M_{s}$,%
\begin{align*}
M_{s}\left[  a\star_{\Omega}(b\star_{\Omega}c)\right]   &  =\sqrt{|\det
s|}(a^{\prime}\star(b\star_{\Omega}c)^{\prime})\\
&  =a^{\prime}\star M_{s}(b\star_{\Omega}c)\\
&  =\sqrt{|\det s|}\left[  a^{\prime}\star(b^{\prime}\star c^{\prime})\right]
.
\end{align*}
A similar calculation yields%
\[
M_{s}\left[  (a\star_{\Omega}b)\star_{\Omega}c)\right]  =\sqrt{|\det
s|}\left[  (a^{\prime}\star b^{\prime})\star c^{\prime}\right]
\]
hence the equality (\ref{massoci}) in view of the associativity of the
Groenewold--Moyal product.
\end{proof}

That we have a deformation of a Poisson bracket follows from the following
considerations. Let us define an $\Omega$-Poisson bracket $\{\cdot
,\cdot\}_{\Omega}$ by%
\begin{equation}
\{a,b\}_{\Omega}= - \omega(X_{a,\Omega},X_{b,\Omega}) \label{poisom}%
\end{equation}
where the vector fields $X_{a,\Omega}$ and $X_{b,\Omega}$ are given by
\begin{equation}
X_{a,\Omega}=\Omega\partial_{z}a\text{ \ , \ }X_{b,\Omega}=\Omega\partial
_{z}b. \label{xab}%
\end{equation}
In particular $\{a,b\}_{\Omega}$ is the usual Poisson bracket $\{a,b\}$ and
$X_{a,J}$, $X_{a,J}$ are the usual Hamilton vector fields when $\Theta=N=0$.
We have the following asymptotic formula relating both notions of Poisson
brackets:
\begin{equation}
\{a,b\}_{\Omega}=\{a,b\}+o(\hbar)\text{ for }\hbar\rightarrow0. \label{abh}%
\end{equation}
In fact, by definitions (\ref{poisom}) and (\ref{xab}),%
\[
\{a,b\}_{\Omega}= - X_{a,\Omega}\cdot\Omega^{-1}X_{b,\Omega}= - \Omega
\partial_{z}a\cdot\partial_{z}b
\]
that is%
\[
\{a,b\}_{\Omega}=\{a,b\} - \hbar^{-1}(\Theta\partial_{x}a\cdot\partial
_{x}b+N\partial_{p}a\cdot\partial_{p}b)
\]
from which (\ref{abh}) follows in view of the conditions (\ref{oh}).

\begin{proposition}
We have
\begin{equation}
a\star_{\Omega}b-b\star_{\Omega}a=i\hbar\{a,b\}+O(\hbar^{2}). \label{poisdef}%
\end{equation}

\end{proposition}

\begin{proof}
We have, since $M_{s}$ is linear,%
\begin{align*}
M_{s}(a\star_{\Omega}b-b\star_{\Omega}a)  &  =\sqrt{|\det s|}(a^{\prime}\star
b^{\prime}-b^{\prime}\star a^{\prime})\\
&  =\sqrt{|\det s|}(i\hbar\{a^{\prime},b^{\prime}\}+O(\hbar^{2}))
\end{align*}
where, as usual, $a^{\prime}(z)=a(sz)$ and $b^{\prime}(z)=b(sz$). Now, by the
chain rule and the relation $Js^{T}=s^{-1}\Omega$,
\begin{align*}
X_{a^{\prime}}  &  =Js^{T}\partial_{z}a(sz)=s^{-1}\Omega\partial_{z}a(sz)\\
X_{b^{\prime}}  &  =Js^{T}\partial_{z}b(sz)=s^{-1}\Omega\partial
_{z}b(sz)\text{ }%
\end{align*}
and hence, using the identities $\{a^{\prime},b^{\prime}\}=JX_{a^{\prime}%
}\cdot X_{b^{\prime}}$ and $(s^{T})^{-1}J^{-1}s^{-1}=\Omega^{-1}$,
\begin{align*}
\sqrt{|\det s|}\{a^{\prime},b^{\prime}\}  &  =-\sqrt{|\det s|}Js^{-1}%
\Omega\partial_{z}a(sz)\cdot s^{-1}\Omega\partial_{z}b(sz)\\
&  =\sqrt{|\det s|}\partial_{z}a(sz)\cdot\Omega\partial_{z}b(sz)\\
&  =-\sqrt{|\det s|}\Omega\partial_{z}a(sz)\cdot\Omega^{-1}(\Omega\partial
_{z}b(sz))\\
&  =-M_{s}\omega(X_{a,\Omega},X_{b,\Omega}).
\end{align*}
We have thus proven that%
\[
M_{s}(a\star_{\Omega}b-b\star_{\Omega}a)=-i\hbar M_{s}\omega(X_{a,\Omega
},X_{b,\Omega})+O(\hbar^{2}).
\]
From this and (\ref{abh}), Eqn. (\ref{poisdef}) follows.
\end{proof}

More generally, using the approach above, it is easy to show that
\[
f\star_{\Omega}g=\sum_{k\geq0}B_{k,\Omega}(f,g)\hbar^{k}%
\]
where the $B_{k,\Omega}(f,g)$ are bi-differential operators. In particular,
$B_{0,\Omega}(f,g)=1$ and $B_{1,\Omega}(f,g)=\frac{i}{2}\left\{  f,g\right\}
$, but for $k\geq2$ they differ from those of the usual Moyal product. We
leave these technicalities aside in this article.

\section{The Intertwining Property}

In \cite{GOLU1} two of us defined a family of partial isometries $W_{\phi
}:L^{2}(\mathbb{R}^{n})\longrightarrow L^{2}(\mathbb{R}^{2n})$ indexed by
$\mathcal{S}(\mathbb{R}^{n}),$ and intertwining the operator $\widetilde
{A}=\widetilde{A}_{J}$ and the usual Weyl operator $\widehat{A}$:%
\begin{equation}
\widetilde{A}W_{\phi}=W_{\phi}\widehat{A}\text{ \ and \ }W_{\phi}^{\ast
}\widetilde{A}=\widehat{A}W_{\phi}^{\ast}\text{.} \label{fund}%
\end{equation}
These intertwiners are defined by%

\begin{equation}
W_{\phi}\psi=(2\pi\hbar)^{n/2}W(\psi,\phi) \label{wpf}%
\end{equation}
where $W(\psi,\phi)$ is the cross-Wigner distribution:%
\begin{equation}
W(\psi,\phi)(z)=\left(  \tfrac{1}{2\pi\hbar}\right)  ^{n}\int_{\mathbb{R}^{n}%
}e^{-\frac{i}{\hbar}p\cdot y}\psi(x+\tfrac{1}{2}y)\overline{\phi(x-\tfrac
{1}{2}y)}dy \label{wifi}%
\end{equation}
and $W_{\phi}^{\ast}$ denotes the adjoint of $W_{\phi}$.

The following result is an extension of Proposition 2 in \cite{GOLU1}.

\begin{theorem}
\label{th1}Let $s$ be a linear automorphism of $\mathbb{R}^{2n}$ such that
$s^{\ast}\omega=\sigma$. (i) The mappings $W_{s,\phi}:\mathcal{S}%
(\mathbb{R}^{n})\longrightarrow S(\mathbb{R}^{2n})$ defined by the formula:%
\begin{equation}
W_{s,\phi}=M_{s}^{-1}W_{\phi} \label{ws}%
\end{equation}
are partial isometries $L^{2}(\mathbb{R}^{n})\longrightarrow L^{2}%
(\mathbb{R}^{2n})$ and we have%
\begin{equation}
\widetilde{A}_{\Omega}W_{s,\phi}=W_{s,\phi}\widehat{A^{\prime}}\text{ \ and
\ }W_{s,\phi}^{\ast}\widetilde{A}_{\Omega}=\widehat{A^{\prime}}W_{s,\phi
}^{\ast} \label{sfund}%
\end{equation}
where $\widehat{A^{\prime}}$ is the operator with Weyl symbol $a^{\prime
}=a(sz)$ and $W_{s,\phi}^{\ast}$ denotes the adjoint of $W_{s,\phi}$. (ii) The
replacement of $s$ by $s^{\prime}$ such that $\sigma=s^{\prime\ast}\omega$ is
equivalent to the replacement of $\widehat{A^{\prime}}$ by $\widehat
{S_{\sigma}}^{-1}\widehat{A^{\prime}}\widehat{S_{\sigma}}$ and of $W_{s,\phi}$
by $W_{s,\widehat{S_{\sigma}^{-1}}\phi}\widehat{S_{\sigma}^{-1}}$ where
$\widehat{S_{\sigma}}$ is any of the two operators in the metaplectic group
$\operatorname*{Mp}(2n,\mathbb{R})$ whose projection on $\operatorname*{Sp}%
(2n,\mathbb{R})$ is $s_{\sigma}=s^{-1}s^{\prime}$.
\end{theorem}

\begin{proof}
(i) We have, using the first formula (\ref{fund}) and definition (\ref{ws}),%
\[
\widetilde{A}_{\Omega}W_{s,\phi}=M_{s}^{-1}\widetilde{A^{\prime}}M_{s}%
(M_{s}^{-1}W_{\phi})
\]
that is,
\[
\widetilde{A}_{\Omega}W_{s,\phi}=M_{s}^{-1}(\widetilde{A^{\prime}}W_{\phi
})=M_{s}^{-1}W_{\phi}\widehat{A^{\prime}}=W_{s,\phi}\widehat{A^{\prime}};
\]
the equality $W_{s,\phi}^{\ast}\widetilde{A}_{\Omega}=\widehat{A^{\prime}%
}W_{s,\phi}^{\ast}$ is proven in a similar way. That $W_{s,\phi}$ is a partial
isometry is obvious since $W_{\phi}$ is a a partial isometry and $M_{s}$ is
unitary. (ii) We have $s_{\sigma}=s^{-1}s^{\prime}\in\operatorname*{Sp}%
(2n,\mathbb{R})$ hence $a(s^{\prime}z)=a(ss_{\sigma}z)=a^{\prime}(s_{\sigma
}z)$. Let $\widehat{A^{\prime\prime}}$ be the operator with Weyl symbol
$a^{\prime\prime}(z)=a^{\prime}(s_{\sigma}z)$. In view of the symplectic
covariance property of Weyl operators we have $\widehat{A^{\prime\prime}%
}=\widehat{S_{\sigma}}^{-1}\widehat{A^{\prime}}\widehat{S_{\sigma}}$.
Similarly,
\begin{align*}
W_{s^{\prime},\phi}\psi(z)  &  =M_{s}^{-1}(M_{s}M_{s^{\prime}}^{-1}W_{\phi
})\psi(z)\\
&  =M_{s}^{-1}W_{\phi}\psi(s_{\sigma}^{-1}z)\\
&  =W_{s,\phi}\psi(s_{\sigma}^{-1}z)
\end{align*}
hence $W_{s^{\prime},\phi}\psi=W_{s,\widehat{S_{\sigma}}\phi}\widehat
{S_{\sigma}}\psi$ in view of the symplectic covariance of the cross-Wigner
transform (\ref{wifi}); the result follows.
\end{proof}

An important property of the mappings $W_{s,\phi}$ is that they can be used to
construct orthonormal bases in $L^{2}(\mathbb{R}^{2n})$ starting from an
orthonormal basis in $L^{2}(\mathbb{R}^{n})$.

\begin{proposition}
\label{OB}Let $(\phi_{j})_{j\in F}$ be an arbitrary orthonormal basis of
$L^{2}(\mathbb{R}^{n})$; the functions $\Phi_{j,k}=W_{s,\phi_{j}}\phi_{k}$
with $(j,k)\in F\times F$ form an orthonormal basis of $L^{2}(\mathbb{R}%
^{2n})$, and we have $\Phi_{j,k}\in\mathcal{H}_{j}\cap\mathcal{H}_{k}$, with
$\mathcal{H}_{j} =W_{s,\phi_{j}} (L^{2} (\mathbb{R}^{n}))$.
\end{proposition}

\begin{proof}
In \cite{GOLU1} the property was proven for the mappings $W_{\phi_{j}%
}=W_{I,\phi_{j}}$; the lemma follows since $W_{s,\phi}=M_{s}^{-1}W_{\phi}$ and
$M_{s}$ is unitary.
\end{proof}

\section{The $\star_{\Omega}$-Genvalue Equation:\ Spectral Results}

Let us consider the star-genvalue equation for the star-product $\star
_{\Omega}$:%
\begin{equation}
a\star_{\Omega}\Psi=\lambda\Psi; \label{stargen1}%
\end{equation}
here $a$ can be viewed as some Hamiltonian function whose properties are going
to be described, and $\Psi$ a \textquotedblleft phase-space
function\textquotedblright. Following definition (\ref{defstar}) the study of
this problem is equivalent to that of the eigenvalue equation%
\begin{equation}
\widetilde{A}_{\Omega}\Psi=\lambda\Psi\label{stargen2}%
\end{equation}
for the pseudo-differential operator $\widetilde{A}_{\Omega}$. Using the
intertwining relations (\ref{sfund}) it is easy to relate the eigenvalues of
$\widetilde{A}_{\Omega}$ to those of $\widehat{A}^{\prime}$ following the
lines in \cite{GOLU1}; for instance one sees, adapting mutatis mutandis the
proof of Theorem 4 in the reference, that the operators $\widetilde{A}%
_{\Omega}$ and $\widehat{A^{\prime}}$ have the same eigenvalues (see Theorem
\ref{theospec} below). Note that it follows from Theorem \ref{th1}(ii) that
the eigenvalues of $\widehat{A^{\prime}}$ do not depend on the choice of $s$
such that $s^{\ast}\omega=\sigma$.

\begin{theorem}
\label{theospec}The operators $\widetilde{A}_{\Omega}$ and $\widehat
{A^{\prime}}$ have the same eigenvalues. (i) Let $\psi$ be an eigenvector of
$\widehat{A^{\prime}}$: $\widehat{A^{\prime}}\psi=\lambda\psi$. Then
$\Psi=W_{s,\phi}\psi$ is an eigenvector of $\widetilde{A}_{\Omega}$
corresponding to the same eigenvalue: $\widetilde{A}_{\Omega}\Psi=\lambda\Psi
$. (ii) Conversely, if $\Psi$ is an eigenvector of $\widetilde{A}_{\Omega}$
then $\psi=W_{s,\phi}^{\ast}\Psi$ is an eigenvector of $\widehat{A^{\prime}}$
corresponding to the same eigenvalue.
\end{theorem}

\begin{proof}
That every eigenvalue of $\widehat{A^{\prime}}$ also is an eigenvalue of
$\widetilde{A}_{\Omega}$ is clear: if $\widehat{A^{\prime}} \psi=\lambda\psi$
for some $\psi\neq0$ then
\[
\widetilde{A}_{\Omega}(W_{s,\phi}\psi)=W_{s,\phi}\widehat{A^{\prime}}%
\psi=\lambda W_{s,\phi}\psi
\]
and $\Psi=W_{s,\phi}\psi\neq0$ ; this proves at the same time that $W_{s,\phi
}\psi$ is an eigenvector of $\widetilde{A}_{\Omega}$ because $W_{s,\phi}$ is
injective. (ii) Assume conversely that $\widetilde{A}_{\Omega}\Psi=\lambda
\Psi$ for $\Psi\in L^{2}(\mathbb{R}^{2n})$, $\Psi\neq0$, and $\lambda
\in\mathbb{R}$. For every $\phi$ we have
\[
\widehat{A^{\prime}}W_{s,\phi}^{\ast}\Psi=W_{s,\phi}^{\ast}\widetilde
{A}_{\Omega}\Psi=\lambda W_{s,\phi}^{\ast}\Psi
\]
hence $\lambda$ is an eigenvalue of $\widehat{A^{\prime}}$ and $\psi$ an
eigenvector if $\psi=W_{s,\phi}^{\ast}\Psi\neq0$. We have $W_{s,\phi}%
\psi=W_{s,\phi}W_{s,\phi}^{\ast}\Psi=P_{s,\phi}\Psi$ where $P_{s,\phi}$ is the
orthogonal projection on the range $\mathcal{H}_{s,\phi}$ of $W_{s,\phi}$.
Assume that $\psi=0$; then $P_{s,\phi}\Psi=0$ for every $\phi\in
\mathcal{S}(\mathbb{R}^{n}),$ and hence $\Psi=0$ in view of Proposition
\ref{OB}.
\end{proof}

Let us give an application of the result above. Assume that the symbol $a$
belongs to the Shubin class $H\Gamma_{\rho}^{m_{1},m_{0}}(\mathbb{R}^{2n})$;
recall \cite{Shubin} that $a\in H\Gamma_{\rho}^{m_{1},m_{0}}(\mathbb{R}^{2n})$
($m_{0},m_{1}\in\mathbb{R}$ and $0<\rho\leq1$) if $a\in C^{\infty}%
(\mathbb{R}^{2n})$ and if there exist constants $C_{0},C_{1}\geq0$ and, for
every $\alpha\in\mathbb{N}^{n}$, $|\alpha|\neq0$, a constant $C_{\alpha}\geq
0$, such that for $|z|$ sufficiently large
\begin{equation}
C_{0}|z|^{m_{0}}\leq|a(z)|\leq C_{1}|z|^{m_{1}}\text{ \ , \ }|\partial
_{z}^{\alpha}a(z)|\leq C_{\alpha}|a(z)||z|^{-\rho|\alpha|}. \label{shu}%
\end{equation}
The following result (Shubin \cite{Shubin}, Chapter 4) is important in our context:

\begin{theorem}
\label{theshu}Let $a\in H\Gamma_{\rho}^{m_{1},m_{0}}(\mathbb{R}^{2n})$ be
real, and $m_{0}>0$. Then the formally self-adjoint operator $\widehat{A}$
with Weyl symbol $a$ has the following properties: (i) $\widehat{A}$ is
essentially self-adjoint in $L^{2}(\mathbb{R}^{n})$ and has discrete spectrum;
(ii) There exists an orthonormal basis of eigenfunctions $\phi_{j}%
\in\mathcal{S}(\mathbb{R}^{n})$ ($j=1,2,...$) with eigenvalues $\lambda_{j}%
\in\mathbb{R}$ such that $\lim_{j\rightarrow\infty}|\lambda_{j}|=\infty$.
\end{theorem}

It follows that:

\begin{theorem}
Let $a\in H\Gamma_{\rho}^{m_{1},m_{0}}(\mathbb{R}^{2n})$ be real, and
$m_{0}>0$. (i) The stargenvalue equation $a\star_{\Omega}\Psi=\lambda\Psi$ has
a sequence of real eigenvalues $\lambda_{j}$ such that $\lim_{j\rightarrow
\infty}|\lambda_{j}|=\infty$, and these eigenvalues are those of the operator
$\widehat{A^{\prime}}$ with Weyl\ symbol $a^{\prime}(z)=a(sz)$. (ii) The
star-eigenvectors of $a$ are in one-to-one correspondence with the
eigenvectors $\phi_{j}\in\mathcal{S}(\mathbb{R}^{n})$ of $\widehat{A^{\prime}%
}$ by the formula $\Phi_{k,j}=W_{s,\phi_{k}}\phi_{j}$.
\end{theorem}

\begin{proof}
It is an immediate consequence of Theorems \ref{theospec} and \ref{theshu}.
\end{proof}

\section{Concluding Remarks...}

The results using the generalized Weyl-Wigner map \cite{Dias1} seem to be
quite general since they also apply to the case of nonlinear transformations
of $\mathbb{R}^{2n}$ (for a review see also section II of \cite{dipra1}). In
particular a more general starproduct than the one of non-commutative quantum
mechanics was obtained in \cite{dipra1} (see Eqn.(23) in this reference). A
future project could be to extend the approach of the present paper to this
case. Another important topic we have not addressed in this article is the
characterization of the optimal symbol classes and function spaces associated
with the star-product $\star_{\Omega}$. As two of us have shown elsewhere
\cite{GOLU2} Feichtinger's modulations spaces (see \cite{fe06,felu06} for a
review) and the closely related Sj\"{o}strand classes \cite{sj94} are
excellent candidates in the case of Landau-type operators (which are a variant
of the operators $\widetilde{A}$ corresponding to the case $\Omega=J$). It
seems very plausible that these function spaces are likely to play an equally
important role in the theory of the star-product $\star_{\Omega}$. Another
future project concerns a discussion of the starproduct
$\star_\Omega$
and its connection to Rieffel's work in deformation quantization as outlined
in \cite{ri93-1}, and the methods introduced in \cite{LuMa} in a different context.

\begin{acknowledgement}
Maurice de Gosson has been financed by the Austrian Research Agency FWF
(Projektnummer P20442-N13). Nuno Costa Dias and Jo\~{a}o Nuno Prata have been
supported by the grant PDTC/MAT/ 69635/2006 of the Portuguese Science
Foundation (FCT). Franz Luef has been financed by the Marie Curie Outgoing
Fellowship PIOF 220464.
\end{acknowledgement}

\textbf{Author's e-mail addresses:}\bigskip

\textbf{Nuno Costa Dias\footnote{ncdias@meo.pt}}

\textbf{Maurice de Gosson\footnote{maurice.de.gosson@univie.ac.at}}

\textbf{Franz Luef\footnote{franz.luef@univie.ac.at}}

\textbf{Jo\~{a}o Nuno Prata\footnote{joao.prata@mail.telepac.pt }}

\textbf{ }

\end{document}